\documentclass[a4paper,10pt]{article}
\usepackage{fullpage}

\usepackage{amsfonts,amssymb,amsmath,xspace,amsthm}

\usepackage{url}
\usepackage{array}
\usepackage{graphicx,tikz}
\usetikzlibrary{shapes,snakes}
\newtheorem{theorem}{Theorem}
\newtheorem{rk}[theorem]{Remark}

\newtheorem{proposition}[theorem]{Proposition}

\newtheorem{ex}[theorem]{Example}
\newtheorem{definition}[theorem]{Definition}
\newtheorem{lemma}[theorem]{Lemma}
\newtheorem{corollary}[theorem]{Corollary}

\newcommand{\TsP}{\ensuremath{\mathbf{T}(\sigma)}\xspace}
\newcommand{\algo}{Algorithm BBCP07}

\begin{document}


%

  \title{Average-case analysis of perfect sorting by reversals}

  \author{Mathilde Bouvel\footnote{CNRS, Universit\'e de Bordeaux I, LaBRI, Bordeaux, France}, Cedric Chauve\footnote{Department of Mathematics, Simon Fraser University, Burnaby (BC), Canada, V5A 1S6}, \addtocounter{footnote}{-1} Marni Mishna\footnotemark, Dominique Rossin\footnote{CNRS, \'Ecole Polytechnique, LIX, Palaiseau, France}}
\date{~}

\maketitle


\begin{abstract}
  Perfect sorting by reversals, a problem originating in computational
  genomics, is the process of sorting a signed permutation to either the
  identity or to the reversed identity permutation, by a sequence of
  reversals that do not break any common interval. B\'erard {\em et
    al.} (2007) make use of strong interval trees to describe an
  algorithm for sorting signed permutations by reversals.
  Combinatorial properties of this family of trees are essential to
  the algorithm analysis. Here, we use the expected value of
  certain tree parameters to prove that the average run-time of the
  algorithm is at worst, polynomial, and additionally, for
  sufficiently long permutations, the sorting algorithm runs in
  polynomial time with probability one.  Furthermore, our analysis of
  the subclass of commuting scenarios yields precise results on the
  average length of a reversal, and the average number of reversals.

\noindent\emph{A preliminary version of this work appeared in the
  proceedings of Combinatorial Pattern Matching (CPM) 2009, Lectures
  Notes in Computer Science, vol. 5577, pp. 314--325, Springer.}
\end{abstract}


\section{Introduction}
\label{sec:intro}
There are many examples where the average case complexity of a sorting
algorithm is neatly computed with a generating function computation
on a related family of trees. Most of the heavy lifting is done by
complex analysis.  We give a new example here: we perform an average
case analysis of a sorting algorithm from computational genomics by
generating function analysis of a family of trees.

\paragraph{Motivation: a computational genomics problem.}
With the availability of a growing number of sequenced and assembled
genomes, the comparison of whole genomes in terms of large-scale
evolutionary events called \emph{genome rearrangements} is a
fundamental task in computational genomics. Computing a genomic
distance and/or a parsimonious evolutionary scenario between a pair of
genomes is one of the basics problems in this field, with applications
such as reconstructing phylogenies~\cite{moret-reconstructing} or
unraveling evolutionary properties of groups of
genomes~\cite{lefebvre,peng-fragile}. This general problem was
formally introduced as an algorithmic problem by Sankoff
in~\cite{sankoff-edit}. Since then, these questions have been
extensively investigated, for different models of genomes and genome
rearrangements, leading to a rich corpus of combinatorial and
algorithmic results; we refer the reader to the recent book by Fertin
\emph{et al.} on this topic~\cite{fertin09}.

\paragraph{Signed permutations, reversals and scenarios.}
In this work, we study the problem of computing parsimonious
perfect reversal scenarios between unichromosomal
genomes. Unichromosomal genomes can be modeled by signed permutations:
each element of a permutation corresponds to a genomic marker (a gene
for example but not exclusively), defined as a segment of the
double-stranded DNA molecule forming a chromosome segment, with its
sign indicating which strand of the chromosome carries the marker.

A reversal is an evolutionary event that reverses a chromosomal
segment. It can be modeled as a discrete operator acting on a signed
permutation, reversing the order and sign of an interval of the
permutation. A sequence of reversals that transforms one signed
permutation into another one is viewed as a possible evolutionary
scenario from a genome to another one. Such a scenario is said to be
\emph{parsimonious} if no other scenario exists that requires less
reversals.

Notice that up to relabeling, we can always assume that one of the two
permutations is the identity. Without loss of generality, we
assume that the permutation we want to obtain at the end of a scenario
is the identity, hence the connection with the sorting
problem. Sankoff initiated in~\cite{sankoff-edit} the algorithmic
study of parsimonious reversal scenarios. Since then this problem has
been considered by many authors, and efficient algorithms exist to
compute a parsimonious
scenario~\cite{hannenhalli-transforming,bergeron-inversion,tannier-advances}.

\paragraph{Common intervals and perfect scenarios.}
However, there can be many scenarios that satisfy this parsimony
constraint. In fact, on real data sets, there can be an exponential
number of parsimonious reversal scenarios (see~\cite{braga-exploring}
for example). This illustrates the need to refine the criteria which
defines a good evolutionary scenario, and to go beyond the simple
parsimony criterion. This need motivates the introduction of the
\emph{perfect scenarios}.

Perfect scenarios aim at avoiding convergent evolution. That is, if
groups of genes or other genomic markers are co-localized in genomes
of two species, a preferred scenario would preserve this quality back
to the ancestral genome; the group of genes should remain together in
every step of the evolution. In the combinatorial model based on
signed permutations, this appears as a \emph{common interval\/}: a
collection of sequential numbers that forms an interval both in the
identity permutation and (in absolute values) in the signed
permutation to be sorted. A sorting scenario for a signed permutation
$\sigma$ is said to \emph{break} a common interval $I$ of $\sigma$
if it contains a reversal such that the elements of $I$ do not form
an interval in the permutation obtained after the reversal is
performed.  A scenario that does not break any common interval is said
to be \emph{perfect}, and may very well be longer than the shortest,
purely parsimonious, scenario. However it is considered to have
stronger properties as a hypothetical evolutionary scenario. The
algorithmic problem is thus stated:

\begin{quote}
  Given a signed permutation, compute a sequence of reversals that
  sorts it towards the identity or reversed identity, does not break
  any of its common intervals, and is shortest among
  all such scenarios.
\end{quote}

Notice that the permutation obtained at the end of the scenario can be
the identity, or the reversed identity, which represents the same
genome but viewed from the other end. 

\paragraph{Computing perfect scenarios: existing results.}
The refined problem that asks to preserve only a predefined subset of
the existing common intervals is
NP-complete~\cite{figeac-sorting}. Even in the general problem, which
considers all common intervals, no algorithms with polynomial
worst-case time complexity are known. However, some fixed parameter
tractable (FPT) algorithms have been
described~\cite{berard-perfect,berard-more}. 

There also exists some classes of signed permutations that define
tractable
instances~\cite{berard-conservation,berard-perfect,diekmann-evolution}.
Among such tractable classes of signed permutations, {\em commuting
  permutations} is the sub-class of signed permutations that can be
sorted by a \emph{commuting scenario}, \emph{i.e.} by a perfect
scenario with the striking trait that the property of being a perfect
scenario is preserved even when the sequence of reversals is reordered
in every possible way. Surprisingly, examples of commuting scenarios
arise in the study of mammalian genome
evolution~\cite{berard-conservation}.

\paragraph{A link with trees and its applications: new results.}
The central combinatorial object in the theory of perfect sorting by
reversals is the ``strong interval tree'' which tracks all common
intervals of a (signed) permutation.  It serves as a guide for the
computation of perfect scenarios and the parameters introduced in the
FPT algorithms described in~\cite{berard-perfect,berard-more} read
naturally in terms of this tree. This link opens the way to a refined
analysis of some of the existing algorithms for perfect sorting by
reversals, which is the purpose of our work.

The two key new results in Section~\ref{sec:prime} are
Theorem~\ref{thm:formePermutation}, which states that for large enough
$n$, with probability $1$, computing a perfect scenario for signed
permutations can be done in time polynomial in $n$ and
Theorem~\ref{thm:avg}, which states that computing a perfect scenario
can be done in polynomial time on average.
Section~\ref{sec:commuting} offers two new results on the average shape of
a commuting scenario: we show that in parsimonious perfect
scenarios for commuting permutations of size $n$, the average number
of reversals is asymptotically $1.2n$, and the average length of a
reversal is asymptotically $1.05 \sqrt{n}$.

We conclude by discussing the relevance of these results, both from
theoretical and applied point of views, and outlining future research.

\section{Preliminaries}
\label{sec:prelim}


We first summarize the combinatorial and algorithmic frameworks for
perfect sorting by reversals. For a more detailed treatment,
in particular for properties of the strong interval tree, we
refer the reader to \cite{berard-perfect}.

\paragraph{Permutations, reversals, common intervals and perfect
  scenarios.}  

A {\em signed permutation\/} of size~$n$ is a permutation of the set of
integers $\{1, 2, \ldots, n\}$ in which each element additionally
has a sign, either positive or negative. For clarity, negative
integers are represented by placing a bar over them and positive signs
are omitted.  We write our permutations in one line notation. For example,
$\sigma=[1\;\overline{3}\;\overline{2}\;5\;4\;6]$ is a signed
permutation of size 6. We denote by $Id_n$ (resp. $\overline{Id_n}$,
$Id_n^m$) the identity (resp. reversed identity, mirrored identity)
permutation, $[1\ 2\dots n]$ (resp. $[\overline{n}\  \dots\ 
\overline{2}\ \overline{1}]$, $[n\ \dots\ 2\ 1]$). When the number
$n$ of elements is clear from the context, we will simply write $Id$, $\overline{Id}$, or $Id^m$.

An {\em interval} $I$ of a signed permutation $\sigma$ of size $n$ is a
segment of adjacent elements of $\sigma$. The {\em content} of $I$ is
the subset of $\{1,\dots,n\}$ defined by the absolute values of the elements of
$I$. Given $\sigma$, an interval is defined by its content and from now,
when the context is unambiguous, we identify an interval with its
content.

The {\em reversal} of an interval of a signed permutation reverses the
order of the elements of the interval, while changing their signs.
The length of a reversal is the number of elements in the interval
that is reversed.  If $\sigma$ is a permutation, we denote by
$\overline{\sigma}$ the permutation obtained by reversing the complete
permutation $\sigma$. A {\em scenario} for $\sigma$ is a sequence of
reversals that transforms $\sigma$ into $Id_n$ or $\overline{Id_n}$.
The {\em length} of such a scenario is the number of reversals it
contains. A scenario of minimal length is a \emph{parsimonious
  scenario}.

\begin{ex} \label{ex:interval} \em Let
  $\sigma=[1\;\overline{4}\;\overline{5}\;2\;\overline{3}\;6]$ be a
  signed permutation of size 6, then
  $\overline{\sigma}=[\overline{6}\;3\;\overline{2}\;
  5\;4\;\overline{1}]$.  Reversing, in $\sigma$, the interval
  $[\overline{5}\;2\;\overline{3}]$, or equivalently the set $\{2,
  3,5\}$, yields the signed permutation
  $[1\;\overline{4}\;3\;\overline{2}\;5\;6]$. Reversing successively
  $\{2,3,4\}$ and $\{3\}$ completes this first reversal to form a
  parsimonious scenario of length $3$.
\end{ex}

A {\em common interval} of a permutation $\sigma$ of size~$n$ is a subset
of $\{1,2, \ldots, n\}$ that is an interval in both $\sigma$ and the identity
permutation $Id_n$. The singletons and the set $\{1, 2, \ldots, n\}$
are always common intervals called {\em trivial common intervals}.

\begin{ex} \label{ex:common-interval}\em
  The common intervals of $\sigma =
  [1\;\overline{3}\;\overline{2}\;5\;4\;6]$ are $\{2,3\}$,
  $\{1,2,3\}$, $\{4,5\}$, $\{4,5,6\}$, $\{2,3,4,5\}$, $\{2,3,4,5,6\}$,
  $\{1,2,3,4,5\}$, $\{1,2,3,4,5,6\}$, and the singletons $\{1\},\{2\},\{3\},\{4\},\{5\},\{6\}$.  
\end{ex}

Two distinct sets (intervals here) $I$ and $J$ {\em commute} if their
contents trivially intersect, that is either $I \subset J$, or $J
\subset I$, or $I \cap J = \emptyset$.  If intervals $I$ and $J$ do
not commute, they {\em overlap}. A scenario $S$ for $\sigma$ is a {\em
  perfect scenario} if no reversal of $S$ breaks any common interval
of $\sigma$, or equivalently \cite{berard-perfect} if every reversal
of $S$ commutes with every common interval of $\sigma$. It is easy to
see that there always exists a perfect scenario for a given signed
permutation. A perfect scenario of minimal length, among all perfect
scenarios, is a {\em parsimonious perfect scenario}.

A permutation $\sigma$ is said to be {\em commuting} if there exists a
scenario for $\sigma$ such that for every pair of reversals the
corresponding intervals commute. Such a scenario is called a {\em
  commuting scenario} and is obviously perfect. It was shown
in~\cite{berard-perfect} that, if a signed permutation can be sorted
by a commuting scenario, then any other perfect scenario for this
signed permutation has the same set of reversals, and conversely every
reordering of the reversals also gives a perfect scenario. This
implies that a commuting scenario is also a parsimonious perfect
scenario.

\begin{ex} \label{ex:commuting} \em Let
  $\sigma=[1\;\overline{3}\;\overline{2}\;5\;4\;6]$ be a signed
  permutation of size 6. The scenario $\{2, 3\}$, $\{4,5\}$, $\{4\}$,
  $\{5\}$ is a commuting scenario, and $\sigma$ is a commuting
  permutation.
\end{ex}

\begin{rk}\label{rem:commuting}\em
  Commuting permutations have been investigated, in connection with
  permutation patterns, under the name of {\em separable} permutations
  \cite{Iba97}.
\end{rk}

\paragraph{The strong interval tree.}

First, we remark that the following definitions are valid for both
signed and unsigned permutations. A common interval~$I$ of a
permutation $\sigma$ is a {\em strong interval} of $\sigma$ if it
commutes with every other common interval of $\sigma$.  The inclusion
order on the set of strong intervals of a permutation of size~$n$
defines an $n$-leaf tree, denoted by \TsP, whose leaves are the
singletons and whose root is the interval containing all elements of
the permutation. We require that the elements of $\{1,2,\ldots, n\}$ appear on
the leaves of \TsP from left to right in the same order they do in
$\sigma$. This implies that the children of every internal vertex of
\TsP are totally ordered, or in other words that \TsP is a plane tree
\emph{i.e.} a tree embedded in the plane. We identify a vertex of \TsP
with the strong interval it represents. If $\sigma$ is a signed
permutation, the sign of every element of $\sigma$ is given to the
corresponding leaves in \TsP. 
Figure~\ref{fig:prime} shows an example of a strong interval
tree.

 \begin{figure*}
    \begin{center}
    
    \begin{tikzpicture}
\draw (10,0) node[draw, rectangle] (l1) {$1,2,3,4,5,6,7,8,9,10,11,12,13,14,15,16,17,18$};
\draw (8,-1.5) node[ellipse,draw] (l21) {$2,3,4,5,6,7,8,9$};
\draw (12,-1.5) node[ellipse,draw] (l22) {$10,11,12,13,14$};
\draw (l1) -- (l21);
\draw (l1) -- (l22);
\draw (7,-3) node[ellipse,draw] (l31) {$2,3,4,5$};
\draw (10,-3) node[rectangle,draw] (l32) {$6,7$};
\draw (12.5,-3) node[rectangle,draw] (l33) {$13,14$};
\draw (15.5,-3) node[rectangle,draw] (l34) {$16,17$};
\begin{scope}[xshift=3cm]
\draw (1,-4.5) node[rectangle,draw] (l41) {$1$};
\draw (1,-5) node {$+$};
\draw (1.75,-4.5) node[rectangle,draw] (l42) {$8$};
\draw (1.75,-5) node {$-$};
\draw (2.5,-4.5) node[rectangle,draw] (l43) {$4$};
\draw (2.5,-5) node {$+$};
\draw (3.25,-4.5) node[rectangle,draw] (l44) {$2$};
\draw (3.25,-5) node {$+$};
\draw (4,-4.5) node[rectangle,draw] (l45) {$5$};
\draw (4,-5) node {$-$};
\draw (4.75,-4.5) node[rectangle,draw] (l46) {$3$};
\draw (4.75,-5) node {$+$};
\draw (5.5,-4.5) node[rectangle,draw] (l47) {$9$};
\draw (5.5,-5) node {$+$};
\draw (6.25,-4.5) node[rectangle,draw] (l48) {$6$};
\draw (6.25,-5) node {$-$};
\draw (7,-4.5) node[rectangle,draw] (l49) {$7$};
\draw (7,-5) node {$+$};
\draw (7.75,-4.5) node[rectangle,draw] (l410) {$12$};
\draw (7.75,-5) node {$+$};
\draw (8.5,-4.5) node[rectangle,draw] (l411) {$10$};
\draw (8.5,-5) node {$+$};
\draw (9.25,-4.5) node[rectangle,draw] (l412) {$14$};
\draw (9.25,-5) node {$-$};
\draw (10,-4.5) node[rectangle,draw] (l413) {$13$};
\draw (10,-5) node {$+$};
\draw (10.75,-4.5) node[rectangle,draw] (l414) {$11$};
\draw (10.75,-5) node {$-$};
\draw (11.5,-4.5) node[rectangle,draw] (l415) {$15$};
\draw (11.5,-5) node {$+$};
\draw (12.25,-4.5) node[rectangle,draw] (l416) {$17$};
\draw (12.25,-5) node {$-$};
\draw (13,-4.5) node[rectangle,draw] (l417) {$16$};
\draw (13,-5) node {$+$};
\draw (13.75,-4.5) node[rectangle,draw] (l418) {$18$};
\draw (13.75,-5) node {$+$};
\end{scope}
\draw (l21) -- (l31);
\draw (l22) -- (l33);
\draw (l21) .. controls +(-2,-1)  and +(0,1) ..  (l42);
\draw (l21) .. controls +(1,-1)  and +(0,1) ..  (l47);
\draw (l21) .. controls +(1.5,-1)  and +(0,1) ..  (l32);
\draw (l1) .. controls +(-6,-1)  and +(0,1) ..  (l41);
\draw (l1) .. controls +(6,-1)  and +(0,1) ..  (l418);
\draw (l1) .. controls +(4.5,-1)  and +(0,1) ..  (l34);
\draw (l1) .. controls +(6.5,-1.5)  and +(0,1) ..  (l415);
\draw (l34) -- (l416);
\draw (l34) -- (l417);
\draw (l33) -- (l412);
\draw (l33) -- (l413);
\draw (l22) .. controls +(1.5,-1)  and +(0,1) ..  (l414);
\draw (l31) -- (l43);
\draw (l31) -- (l44);
\draw (l31) -- (l45);
\draw (l31) -- (l46);
\draw (l32) -- (l48);
\draw (l32) -- (l49);
\draw (l22) -- (l410);
\draw (l22) -- (l411);

    \end{tikzpicture}
   \end{center}
   
   \caption{%
{\small\em 
The strong interval tree $\mathbf{T}([1~\overline{8}~4~2~\overline{5}~3~9~\overline{6}~7~12~10~\overline{14}~13~\overline{11}~15~\overline{17}~16~18])$.
    Vertices are labeled by the strong intervals.  There are three
    non-trivial linear vertices (rectangular) and three
    prime vertices (round). The root and the vertex
    $\{6,7\}$ are increasing linear vertices, while the linear
    vertices $\{16,17\}$ and $\{13,14\}$ are decreasing.  }}
   \label{fig:prime}
\end{figure*}
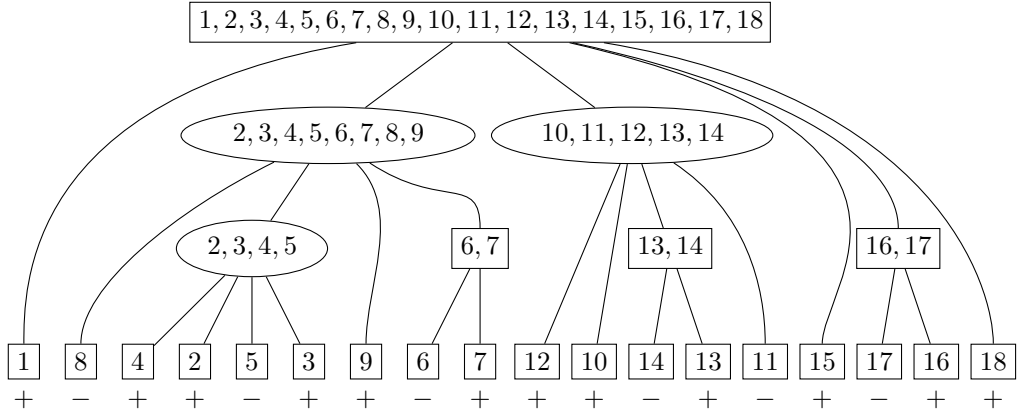

Let $I$ be a strong interval of $\sigma$ that is not a singleton and let
${\cal I}=(I_1,\dots, I_k)$ the unique partition of the elements of
$I$ into maximal strong intervals, from left to right. The
\emph{quotient permutation} of $I$, denoted $\sigma_I$, is the
permutation of size $k$ defined as follows: $\sigma_I(i)$ is smaller than
$\sigma_I(j)$ in $\sigma_I$ if and only if any element of the content of $I_i$ is
smaller than any element of the content of $I_j$. A fundamental property of the strong interval
tree is that the quotient permutation $\sigma_I$ of an internal vertex
$I$ having $k$ children ($k \geq 2$) in a strong interval tree can
only be either $Id_k$, ${Id_k^m}$, or a permutation of size $k$, with $k
\geq 4$, whose only common intervals are the $k+1$ trivial common
intervals. Such a permutation with no non-trivial common interval is
called a {\em simple permutation}. The shortest simple permutations
are of size $4$ and are $[3\ 1\ 4\ 2]$ and $[2\ 4\ 1\ 3]$.  We
describe simple permutations in more detail in
Section~\ref{ssec:combPre}.

For an internal vertex $I$, if $\sigma_I=Id_k$ (resp. ${Id_k^m}$, is a
simple permutation), then $I$ is said to be an {\em increasing linear}
vertex (resp. {\em decreasing linear} vertex, {\em prime}
vertex). Another crucial property of a strong interval tree is that no
two increasing (resp.  decreasing) linear vertices can be adjacent: if
a linear vertex is the child of another linear vertex, then one of
them is increasing and the other one is decreasing.

The strong interval tree is also known as the {\em substitution
  decomposition tree} \cite{AA05}, and is very similar to
\emph{$PQ$-trees} \cite{booth-testing}, a data structure used to
represent the common intervals of two or more
permutations~\cite{heber-finding, bergeron-computing, buixuan-revisiting}. More
precisely, the strong interval tree defines a PQ-tree if linear
(resp. prime) vertices are called Q-vertices (resp. P-vertices). This
PQ-tree can be computed in linear time \cite{bergeron-computing}. To
obtain the strong interval tree, the quotient permutation of each
vertex needs then to be computed. The algorithm
of~\cite{bergeron-computing} can be adapted to compute them, still in
linear time. Indeed, given the tree, the quotient permutations can be
computed as follows: consider the elements on the leaves, from $1$ to
$n$, and propagate these elements along the edges of the tree towards
the root, until a previously used edge is encountered. The relative
ordering of the elements at every internal vertex of the tree gives
the quotient permutation, and their computation is obtained in
$\mathcal{O}(n)$ time.



\smallskip
\paragraph{The strong interval tree as a guide for perfect sorting by
  reversals.}

The algorithm in~\cite{berard-perfect} computing a parsimonious
perfect scenario for a given signed permutation is the central object
of study here, and is henceforth labeled \algo.

To compute a parsimonious perfect scenario for a signed permutation
$\sigma$, \algo\, heavily relies on the strong interval tree \TsP  of
$\sigma$. It starts with computing this tree, and then assign signs to
internal vertices according to the following rules: an increasing
(resp. decreasing) linear vertex is signed $+$ (resp. $-$) and a prime
vertex having a linear parent inherits its sign from its parent. Some
prime vertices may remain unsigned at this step, and the algorithm
will explore all the possible assignments of signs to these prime
nodes. If $p$ denote the number of prime vertices in \TsP, there may
be up to $2^p$ possible assignments. The key ingredient of the
algorithm is that any reversal in a perfect scenario is either a
strong interval (hence a vertex of \TsP) or the union of consecutive
children of a prime vertex of \TsP~\cite[Proposition
  2]{berard-perfect}. Hence a scenario can be computed by looking
successively at each vertex of the strong interval tree, sorting a
signed permutation, defined from its quotient permutation and the
signs of its children, towards either the identity (if the vertex has
sign $+$) or the reversed identity (if the vertex has sign $-$). More
precisely, for each assignment of signs, a scenario is computed as
follows:
\begin{itemize}
\item Transform the quotient permutation of each vertex into a signed
  permutation by lifting the sign of each child onto the
  corresponding element in the quotient permutation.
 \item For each prime node signed $+$ (resp. $-$) whose signed
   quotient permutation is $\tau$ (a signed permutation of size $k$),
   compute a parsimonious scenario from $\tau$ to $Id_k$ (resp. to
   $\overline{Id_k}$). This is achieved using a polynomial-time
   algorithm solving the general sorting by reversal problem (without
   the 'perfectness' condition). The most efficient algorithm so far
   is the one of \cite{tannier-advances}, that runs in $\mathcal{O}(k
   \sqrt{k \log k})$ time.
 \item In addition to the reversals obtained at the previous step,
   perform a reversal for every interval of $\sigma$ that correspond
   to a vertex (internal vertex or leaf) in \TsP whose parent is
   linear and whose sign is different from the sign of its parent.
\end{itemize}

The scenarios thus obtained are all perfect scenarios, and among them,
those of minimal length are parsimonious perfect scenarios. For the
correctness and complexity analysis of \algo, we refer to
\cite{berard-perfect}.

\begin{ex}
On the example of Figure \ref{fig:prime}, the root of \TsP, its two
prime children and vertex $\{6,7\}$ are signed $+$, whereas vertices
$\{13,14\}$ and $\{16,17\}$ are signed $-$. For vertex
$\{2,3,4,5\}$, the two possible sign assignments have to be
tested. 
Choosing sign $+$ (resp. $-$) produces a scenario with $15$
(resp. $14$) reversals, among which $4$ correct a sign mismatch between a vertex and its linear parent (for vertices $\{6\}, \{13\}, \{16\}$ and $\{16,17\}$) and the remaining $11$ (resp. $10$) arise from reversals in prime nodes. More precisely, sorting the right-most prime child of the root requires $3$ reversals (through the optimal scenario $[\mathbf{3\,1\,\overline{4}}\,\overline{2}]\rightarrow [\mathbf{4\,\overline{1}}\,\overline{3}\,\overline{2}] \rightarrow [1\,\mathbf{\overline{4}\,\overline{3}\,\overline{2}}] \rightarrow [1\,2\,3\,4] $); when sign $+$ is chosen, the left-most prime child of the root is sorted in $4$ reversals ($[\mathbf{\overline{3}\,1}\,4\,2]\rightarrow [\mathbf{\overline{1}\,3\,4}\,2] \rightarrow [\mathbf{\overline{4}\,\overline{3}\,1\,2}] \rightarrow [\mathbf{\overline{2}\,\overline{1}}\,3\,4] \rightarrow [1\,2\,3\,4] $) and its prime child in $4$ reversals ($[3\,1\,\mathbf{\overline{4}\,2}]\rightarrow [\mathbf{3\,1}\,\overline{2}\,4] \rightarrow [\overline{1}\,\mathbf{\overline{3}\,\overline{2}}\,4] \rightarrow [\mathbf{\overline{1}}\,2\,3\,4] \rightarrow  [1\,2\,3\,4]$); and when sign $-$ is chosen, the left-most prime child of the root is sorted in $3$ reversals ($ [\overline{3}\,\mathbf{\overline{1}\,4}\,2] \rightarrow [\overline{3}\,\mathbf{\overline{4}\,1\,2}] \rightarrow [\mathbf{\overline{3}\,\overline{2}\,\overline{1}}\,4] \rightarrow  [1\,2\,3\,4]$) and its prime child in $4$ reversals ($[3\,\mathbf{1\,\overline{4}}\,2]\rightarrow [\mathbf{3\,4}\,\overline{1}\,2] \rightarrow [\overline{4}\,\overline{3}\,\mathbf{\overline{1}\,2}] \rightarrow [\overline{4}\,\overline{3}\,\overline{2}\,\mathbf{1}] \rightarrow [\overline{4}\,\overline{3}\,\overline{2}\,\overline{1}]$).
Therefore, for the signed permutation $\sigma$ of Figure \ref{fig:prime}, the length
of a parsimonious perfect scenario is $14$.
\end{ex}

The following proposition is a summary of some of the key results of
\cite{berard-perfect} on \algo, that will play a central role in our
work.

\begin{proposition}[B\'erard {\em et al.}~\cite{berard-perfect}]
\label{thm:complex}
 Let $\sigma$ be a signed permutation of size~$n$.  Let \TsP be its
 strong interval tree, and denote by $p$ its number of prime nodes.
 Then the followings are true:
 \begin{enumerate}
\item \algo\mbox{} compute a parsimonious perfect scenario for
  $\sigma$ in worst-case time $\mathcal{O}(2^p\,n\sqrt{n \log n})$;
\item $\sigma$ is a commuting permutation if and only if $p = 0$;
\item if $\sigma$ is a commuting permutation, then a sorting scenario
  for $\sigma$ is perfect if and only if it consists of one reversal
  for every interval corresponding to a vertex of \TsP that has a sign
  different from its parent.
\end{enumerate}
\end{proposition}

Hence it appears that prime vertices of the strong interval tree are
fundamental in the exponential worst-case behavior of \algo, and more
generally in the hardness of the problem of perfect sorting by
reversals. Indeed, an interpretation of the hardness result given
in~\cite{figeac-sorting} in terms of strong interval tree is that
perfect sorting by reversals is NP-complete for signed permutations whose
strong interval tree contains only prime nodes.

\section{On the number of prime vertices}
\label{sec:prime}

As we shall soon see, the average-time complexity of \algo\ can also
be bounded with the aid of strong interval trees. We use enumerative
results on simple permutations to determine the ``average shape'' of a
tree with $n$ leaves. This average shape is extremely simple and has a
single prime node. From this we can easily bound the average-time
complexity.


\subsection{Combinatorial preliminaries: strong interval trees and
  simple permutations}\label{ssec:combPre} The following formal
description of the underlying structure of the strong interval trees
is useful for our enumerative analysis. 

\begin{definition}\label{def:strong}\em 
  Let ${\cal T}_n$ be the family of plane trees satisfying the
  following properties:
  \begin{itemize}
  \item[P1.] each tree has $n$ leaves ($n$ is the \emph{size} of the
    trees of ${\cal T}_n$);
  \item[P2.] each leaf is labeled by $+$ or $-$;
  \item[P3.] the children of each internal vertex are totally ordered;
  \item[P4.] each internal vertex has at least two children;
  \item[P5.] if an internal vertex has $k$ children, it is labeled
    either by $Id_k$, or ${Id_k^m}$, or a simple permutation of size
    $k$ if $k\geq 4$;
  \item[P6.] no edge is incident to two vertices labeled by $Id$ or
    two vertices labeled by ${Id^m}$.
  \end{itemize}
\end{definition}
We previously noted that each permutation corresponds to a strong
interval tree. We prove next that this correspondence is bijective.

\begin{theorem}\label{thm:bijection}
  There is a bijection between the set of signed permutations of size~$n$
  and ${\cal T}_n$.
\end{theorem}

\begin{proof}
  First, it is immediate to see that a unique tree of ${\cal T}_n$ can be
  obtained from a signed permutation $\sigma$ of size~$n$. Indeed, it is
  enough to modify its strong interval tree \TsP by labeling each leaf
  representing an element of $\sigma$ by its sign, and each internal
  vertex corresponding to a strong interval $I$ by the quotient
  permutation $\sigma_I$.

  To get a signed permutation $\sigma_T$ from a tree $T$ of ${\cal
    T}_n$, we assign signed integers to its leaves and $\sigma_T$ will
  be obtained by reading the leaves from left to right. The absolute
  values of the integers labeling the leaves are obtained by a
  top-down approach. We first assign the set of integers $I= \{1,
  \ldots n\}$ to the root, together with a variable $m$ set to $1$
  indicating the minimal value of $I$. We propagate this assignment
  from the root to the leaves as follows. Consider a node labeled by a
  permutation $\tau$ with $k$ children rooting subtrees of sizes $s_1, \ldots, s_k$
  from left to right, that has been assigned the set $I$ of consecutive
  integers and the variable $m = \min(I)$. Then assign sets $I_1,
  \ldots, I_k$ and variables $m_1, \ldots, m_k$ to its children so
  that $m_i = m+ \sum_{j : \tau(j)<\tau(i)} s_j$ and $I_i = \{m_i,
  \ldots, m_i+s_i-1\}$. At the end of this process, every leaf is
  labeled by an integer $m$ and a set $I=\{m\}$. The signed integer
  assigned to such a leaf is then either $m$ if the leaf has label $+$
  in $T$ or $-m$ if it has label $-$. Notice that the sets $I$
  assigned to the nodes of $T$ actually correspond to the strong
  intervals of $\sigma_T$, ensuring that the above mapping is a bijection.

\end{proof}

Recall that simple permutations are the permutations that have no
non-trivial common interval, and are used here as quotient
permutations of prime nodes. The enumeration of simple permutations
was investigated in \cite{albert-enumeration}. The authors prove that
this enumerative sequence is not P-recursive and there is no known
closed formula for the number of simple permutations of a given size.
Nonetheless they are able to compute a complete asymptotic expression
for the number of simple permutations of size $n$.

\begin{theorem}[Albert {\em et al.}~\cite{albert-enumeration}]\label{thm:simple}
  Let $s_n$ be the number of simple permutations of size $n$. Then
  \begin{equation} 
    s_n = \frac{n!}{e^2} \left(1 -\frac{4}{n}
    +\frac{2}{n(n-1)} +\mathcal{O}\left(\frac{1}{n^3}\right)\right) \text{ when } n
    \rightarrow \infty.
    \label{eq:sn}\end{equation}
\end{theorem}

\subsection{Average shape of strong interval trees}

A \emph{twin} in a strong interval tree is a vertex of degree $2$
such that each of its two children is a leaf. Thus, a twin is a linear
vertex. 

Let us notice that all results in this section apply both to signed
permutations and unsigned permutations: the two mains reasons for it
are that the definition of intervals in a permutation ignores the
signs of the elements, and that $2^n$ signed permutations are
associated to any unsigned permutation $\sigma$ of size~$n$, and this
number does not depend on $\sigma$.

We first state the main result of this section.

\begin{theorem}\label{thm:formePermutation}
  Asymptotically, with probability $1$, a random permutation of size
  $n$ has a strong interval tree of the form:
  \begin{itemize}
  \item the root is a prime vertex;
  \item every child of the root is either a leaf or a twin.
  \end{itemize}
  Moreover, the probability distribution of the number $k$ of twins is
  given by: $P(k) = \frac{2^k}{e^2 k!}$. Consequently, the expected
  number of twins is $2$.
\end{theorem}

Before proving this result, we can notice that it overlaps with
previous results on the expected number of common intervals in
permutations. In their paper introducing the problem of computing the
common intervals of a permutation~\cite{uno-fast}, Uno and Yagiura
showed that the expected number of common intervals of length $2$ in a
permutation is $2-2/n$ and that, for all $\ell > 2$, the expected
number of common intervals of size $\ell$ is $0$ for $n$ large
enough. This implies immediately our result on the shape of the strong
interval tree. Later, Corteel~\emph{et~al.} showed
in~\cite{corteel-common} that the probability distribution of the
number of common intervals of size $2$ follows a Poisson law, with
mean $2$, a result already proved by Kaplanski, in relation with runs
in permutations~\cite{kaplanski-asymptotic}. A similar result was also
proved independently
in~\cite{XBS08}. Theorem~\ref{thm:formePermutation} gathers all these
results together, expressed in terms of the strong interval
tree. Moreover, the proof we give here is new and relies on
enumerative results on simple permutations.

The proof of Theorem~\ref{thm:formePermutation} follows from
Lemma~\ref{lem:albert-enumeration-generalized} below and Theorem
\ref{thm:simple}.


\begin{lemma} \label{lem:albert-enumeration-generalized}
  If $p_{n,k}$ denotes the number of unsigned \footnote{For signed
    permutations, the denominator $n!$ should be replaced by $2^n
    n!$.} permutations of size $n$ which contain a common interval $I$
  of length $k$ then for any fixed positive integer $c$:
  $$\sum_{k=c+2}^{n-c} \frac{p_{n,k}}{n!} \in \mathcal{O}(n^{-c}).$$
\end{lemma}

\begin{proof}
  The proof of Lemma~\ref{lem:albert-enumeration-generalized} is
  essentially identical to Lemma 7 of~\cite{albert-enumeration}:
  We have $p_{n,k} \leq (n-k+1) k!  (n-k+1)!$. Indeed, the right-hand
  side counts the number of quotient permutations corresponding to $I$
  (which is $k!$), the possible values of the minimal element of $I$
  ($n-k+1$) and the structure of the rest of the permutation with one
  more element for the insertion of $I$ ($(n-k+1)!$).  Only the
  extremal terms of the sum can have magnitude ${\mathcal O}(n^{-c})$
  and the remaining terms have magnitude ${\mathcal O}(n^{-c-1})$.
  Since there are fewer than $n$ terms the result of Lemma
  \ref{lem:albert-enumeration-generalized} follows.
\end{proof}

\begin{proof}[Proof of Theorem~\ref{thm:formePermutation}]
  Lemma \ref{lem:albert-enumeration-generalized} with $c = 1$ gives
  that the proportion of non-simple permutations with at least one
  common interval of size greater than or equal to $3$ is
  $\mathcal{O}(n^{-1})$. But permutations whose common intervals are
  only of size $1,2$ or $n$ are exactly permutations whose strong
  interval tree has a prime root and every child is either a leaf or a
  twin.

  Similarly, the number of permutations whose strong interval tree has
  the form of a prime root with $k$ twins is $s_{n-k} {{n-k} \choose
    k} 2^k$.  Given the asymptotic estimate of $s_n$ in
  Equation~(\ref{eq:sn}), we compute the asymptotic estimate for the
  number of such permutations to be~$\frac{n!2^k}{e^2 k!}$, proving
  Theorem \ref{thm:formePermutation}.
\end{proof}

This result has an immediate corollary in terms of perfect sorting by
reversals: the probability that a signed permutation corresponds to an
instance that requires an exponential time computation to be solved
tends to $0$ as $n$ grows.

\begin{corollary}
  \label{thm:complexity} \algo\ runs in $\mathcal{O}(n \sqrt{n\log n})$
  time with probability $1$ as $n \rightarrow \infty$.
\end{corollary}


\subsection{Average  time complexity of perfect sorting by reversals}
Further analysis of the tree family ${\cal T}_n$ yields a polynomial
bound on the average-time complexity of
\algo\ (Theorem~\ref{thm:avg}).

Consider the following sum, which is central in the description of the
complexity of the algorithm:
\[P_n = \frac{1}{T_n}\sum_p 2^p T_{n,p}.\]
Here $T_n$ is the number of strong interval trees with $n$ leaves
($T_n=|{\cal T}_n|=2^nn!$ from Theorem~\ref{thm:bijection}) and
$T_{n,p}$ is the number of such trees with $p$ prime vertices.  The
key step in the algorithm complexity result is essentially reduced to
showing $P_n\in \mathcal{O}(1)$.

As an intermediate step, we find a bound on $U_{n,p}$, the number of
\emph{unsigned} permutations of size $n$ whose strong interval trees
contain $p$ prime vertices, when $p \geq 2$.

\begin{lemma}
  \label{lem:avgu}
  The number $U_{n,p}$ of unsigned permutations of size $n$ whose
  strong interval trees contain $p$ prime vertices with $p \geq 2$ is
  at most $ 48 \frac{(n-1)!}{2^p}$.
\end{lemma}

\begin{proof}
  We proceed by induction on the number $p$ of prime vertices. The
  hypothesis is the following:

\centerline{$({\mathcal H_p}): \forall n, U_{n,p} \leq 48
  \frac{(n-1)!}{2^p}$.}

The hypothesis $({\mathcal H_p})$ is trivially true for $n < 3p+1$,
since a tree containing $p$ prime vertices has at least $3p+1$ leaves.
We initiate the proof with $p=2$ assuming $n \geq 7$. A tree of size
$n$ with two prime vertices can always be decomposed, although not
uniquely, as a tree $T_1$ that contains one prime vertex, where one
leaf is chosen and expanded by a second tree $T_2$ with one prime
vertex. Hence $|T_1|+|T_2| = n+1$. Without loss of generality, one can
assume that the root of $T_2$ is its only prime vertex.  Recall that
the number of trees with one prime vertex with $k$ leaves is at most
$k!$, as such trees are in bijection with a subset of unsigned
permutations of size $k$.  Hence,
\begin{eqnarray*}
 U_{n,2} & \leq &   \sum_{k=4}^{n-3} k! k (n+1-k)! 
 \leq (n+1)! \sum_{k=4}^{n-3} \frac{k}{{{n+1} \choose k}}\\
& \leq &  \frac{(n+1)!}{{{n+1} \choose 4}} \sum_{k=4}^{n-3} k  
 \leq \frac{24 (n+1)!}{(n+1)n(n-1)(n-2)} \sum_{k=0}^{n-3} k\\
& \leq &  \frac{24 (n-1)!}{(n-1)(n-2)}\frac{(n-3)(n-2)}{2} 
 \leq 48 \frac{(n-1)!}{2^2} 
\end{eqnarray*}


Let us now suppose $({\mathcal H_p})$ true and prove $({\mathcal
  H_{p+1}})$. We proceed as before. Indeed, a tree with $p+1$ prime
vertices can be decomposed -- not necessarily uniquely -- as a tree
$T_1$ with $p$ prime vertices, one leaf of which is expanded by
another tree $T_2$ with one prime vertex. As explained before, we can
assume that $n \geq 3(p+1)+1$. Hence:
\begin{eqnarray*}
 U_{n,p+1} & \leq &  \sum_{k=3p+1}^{n-3} U_{k,p} k (n+1-k)!\\
& \leq &  \frac{48}{2^p} \sum_{k=3p+1}^{n-3} (k-1)!k(n+1-k)!\\
& \leq &  \frac{48 (n+1)!}{2^p} \sum_{k=3p+1}^{n-3} \frac{1}{{{n+1} \choose k}}\\
& \leq &  \frac{48 (n+1)!}{2^p} \left[ \frac{1}{{{n+1} \choose {n-3}}} + \sum_{k=3p+1}^{n-4} \frac{1}{{{n+1} \choose k}} \right]\\
& \leq & \frac{48 (n+1)!}{2^p} \left[ \frac{1}{{{n+1} \choose {4}}} + (n-4-3p) \frac{1}{{{n+1} \choose 5}} \right]\\
\end{eqnarray*}
A straightforward analysis by successive derivations on $n$ shows that $\left[\frac{1}{{{n+1} \choose {4}}} + (n-4-3\cdot 2) \frac{1}{{{n+1} \choose 5}}\right] -\frac{1}{2n(n+1)}\leq 0$ for all $n \geq 10$. Hence, since $p \geq 2$, we deduce that $\left[\frac{1}{{{n+1} \choose {4}}} + (n-4-3p) \frac{1}{{{n+1} \choose 5}}\right] \leq \frac{1}{2n(n+1)}$ for all $n \geq 3(p+1)+1$. This ensures that $U_{n,p+1} \leq \frac{48(n-1)!}{2^{p+1}}$ and concludes the proof.
\end{proof}

In the context of signed permutations, Lemma \ref{lem:avgu} immediately yields the following result:

\begin{lemma}
  \label{lem:avg}
  The number $T_{n,p}$ of signed permutations of size $n$ whose strong
  interval trees contain $p$ prime vertices with $p \geq 2$ is at most
  $2^n 48 \frac{(n-1)!}{2^p}$. 
\end{lemma}

\begin{theorem}\label{thm:avg}
  Computing a shortest perfect scenario for a random signed
  permutation can be done with average time complexity bounded by
  $\mathcal{O}(n\sqrt{n\log n})$.
\end{theorem}

\begin{proof}
  First we bound $P_n$. For all $n$, by Lemma~\ref{lem:avg},
\begin{eqnarray*}
     P_n &=& \frac{\sum_p 2^p T_{n,p}}{T_n} \\
           &\leq& \frac{\left(T_{n,0} + 2T_{n,1} +  \sum_{p=2}^{n} 2^n48 (n-1)! \right)}{T_n} \\
           &\leq&  3 + \sum_{p=2}^{n}\frac{48}{n}\\
           &=&  3 + 48(1-\frac{1}{n})
\end{eqnarray*}
  Thus, $P_n \in \mathcal{O}(1)$.  The average time complexity of
  \algo\ for permutations of size $n$ is given by the following sum,
  for some constant $C$:
\[  \frac{C}{T_n} \sum_{p=0}^{n} T_{n,p} 2^p n \sqrt{n\log n}= C\, P_n \, n\sqrt{n\log n}.\]
The result follows since $P_n \in \mathcal{O}(1)$.

\end{proof}


\section{Properties of commuting scenarios} 
\label{sec:commuting}
We observed in the previous section that the typical shape of the
common interval tree associated to a random permutation is very
particular, and it is reasonable to ask if a signed permutation
selected uniformly at random adequately represents the expected shape
of an evolutionary scenario. Experimentally, the strong interval trees
that arise when comparing pairs of mammalian genomes contain few prime
nodes, labeled with small, simple permutations. Rather, they contain
large subtrees with no prime nodes. These subtrees represent commuting
scenarios. At present we are unaware of a weighting operator on signed
permutations which correlates to the probability that such a
permutation could represent an evolutionary scenario on real
data. Indeed, such an operator would greatly aid in determining
realistic run-times for algorithms on biological data and other
properties of evolutionary scenarios. Towards this goal we begin by
investigating the class of strong interval trees with no prime
nodes. These correspond to commuting scenarios.

The trees that represent commuting scenarios are 
particularly well-studied. They fall into the category of \emph{simple
varieties of trees}, and as such, many formulas exist to compute
quantities such as the asymptotic number of trees with~$n$ leaves, and
also distributions associated to various tree parameters. Some of
these parameters have direct relevance to the evolutionary scenario
interpretation. Chapter~\cite[Section VII.3]{flajolet-analytic} is a
pedagogical reference for simple varieties of trees and we
outline how to derive some key values here. 

In the remainder of the section, we prove the following results on
parsimonious perfect scenarios sorting a commuting signed permutation
of size~$n$, via common interval trees:
\begin{enumerate}
\item The asymptotic number of commuting permutations is
  $2^{n+1} \cdot 0.12\, (5.88)^{n}\,n^{-3/2}$ (a very typical expression for trees) (Equation~\ref{eqn:numtrees});
\item The average number of reversals in one of these scenarios is
  $1.2\,n$ (Theorem~\ref{thm:numreversals}). This
  is a consequence of the average number of internal vertices in the tree
  (Equation~\ref{eqn:avg-iv});
\item The average length of a reversal is $1.05\, \sqrt{n}$. This is
  related to the average pathlength of a tree. (Theorem~\ref{thm:avglen})
\end{enumerate}

Additionally, in the proof of Theorem~\ref{thm:numreversals}, we
  can determine that 37\% of the expected reversals have length
  1. This agrees with the observation of a large proportion of short
  reversals in parsimonious scenarios for bacterial
  genomes~\cite{lefebvre}.

Finally, a note on convergence. The asymptotic estimates we present
converge quickly, even for relatively small $n$. For example, the
estimate given for the number of commuting permutations is correct up
to order $O(n^{-5/2})$. In real terms, at $n=100$ it is within 3\% of
the real value.  The parameters have a similar accuracy. The trees
that arise from biological data have on the order of 1000 leaves
(see~\cite{landau-gene} for example), and hence these are very strong
estimates.

\subsection{Modified Schr\"oder Trees}
Let $\sigma$ be a commuting permutation of size $n$, equivalently, a
signed permutation whose strong interval tree \TsP has no prime
node. Thus, $\TsP$ is a plane tree with the property that the internal
vertices have at least two children, each leaf is signed either~$+$
or~$-$, and the root is also signed~$+$ or~$-$ (to indicate whether it is an 
increasing or a decreasing linear vertex). The signs of the other
internal vertices follow unambiguously from the sign of the root, alternating 
between $+$ and $-$ along each branch of the tree.

Disregarding the signs on the leaves and root, this family of trees is
known as \emph{Schr\"oder trees} (entry A001003 in the On-Line
Encyclopedia of Integer Sequences \cite{njas}), and they are
straightforward to analyze. 

Let $\mathcal{C}$ be the class of
all strong interval trees representing commuting permutations, and let
$\mathcal{S}$ be the class of Schr\"oder trees. If $C_n$ and
$S_n$ respectively denote the number of trees with $n$ leaves in these
two classes, then
\[
C_n= 2\cdot 2^n\cdot S_n.
\]
Because of this exact~$\{1:2^{n+1}\}$ correspondence, we generally 
first consider the class~$\mathcal{S}$ to determine structural properties,
and then account for the contribution from the leaves. We remark that
$\mathcal{S}$ is a subset of the trees $\mathcal{T}$.

\subsection{A specification for $\mathcal{S}$}
Like many tree classes, there is a simple recursive description for
the class $\mathcal{S}$ of Schr\"oder trees: A tree is either
a leaf (denoted $\mathcal{L}$), or an internal vertex with at least
two subtrees, all of which are elements of $\mathcal{S}$. A visual representation 
of this statement is given on Figure \ref{fig:schroder}.

\begin{figure}[ht]
\center
\begin{tikzpicture}
\fill (0.5,0.5) circle (1.6pt) node [below=1.5pt] {$\mathcal{S}$};
\draw [semithick] (0,0) -- (1,0) -- (0.5,0.5) --(0,0);
\end{tikzpicture} =
\begin{tikzpicture}
\draw [semithick] (0,0) -- (.5,0) -- (.5,.5) --(0,.5) -- (0,0);
\draw (0.25,.25) node {$\mathcal{L}$};
\end{tikzpicture}  +
\begin{tikzpicture}
\fill (2,1) circle (1.6pt) node {};

\fill (0.5,0.5) circle (1.6pt) node [below=1.5pt] {$\mathcal{S}$};
\draw [semithick] (0.5,0.5) --(0,0)  -- (1,0) -- (0.5,0.5) --(2,1);

\fill (2,0.5) circle (1.6pt) node [below=1.5pt] {$\mathcal{S}$};
\draw [semithick] (2,0.5) --(1.5,0) -- (2.5,0) -- (2,0.5) --(2,1);

\fill (3.7,0.5) circle (1.6pt) node [below=1.5pt] {$\mathcal{S}$};
\draw [semithick] (3.7,0.5) --(3.2,0) -- (4.2,0) -- (3.7,0.5) --(2,1);

\draw (2.85,0.3) node {$\dots$};
\end{tikzpicture}
\caption{A Schr\"oder tree is decomposed as either a leaf, or an internal vertex
  with a sequence of subtrees. }
\label{fig:schroder}
\end{figure}
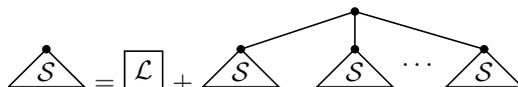
We say that the size of a tree $\tau\in\mathcal{S}$ is the number of
leaves, and we denote this quantity by $|\tau|$. We shall later
consider the number of internal vertices. A leaf is 
an atomic structure of weight one, and an
internal vertex is a neutral structure of weight 0.  We translate the
above picture description of $\mathcal{S}$ into the following
combinatorial equation:
\begin{equation}\label{eqn:schroder}
\mathcal{S} = \mathcal{L} + \operatorname{Seq}_{\geq 2}(\mathcal{S}),
\end{equation}
where $\operatorname{Seq}_{\geq 2}(\mathcal{S})$ represents a
\emph{sequence} (total order) of at least two trees of $\mathcal{S}$.

\subsection{A crash course on decomposable structures}
The formalism we are using here is well described
in~\cite{flajolet-analytic}. The main advantage is the direct access to
functional equations for the ordinary generating functions. Recall, if
$S_n$ is the number trees in $\mathcal{S}$ with $n$ leaves, the
ordinary generating function (ogf) is defined
as the formal power series~$S(z)=\sum_n S_n z^n$.  Thus, the series
expansion of~$S(z)$ 
begins $S(z)=z+z^2+3z^3+11z^4+\cdots$.  Many combinatorial
actions described on combinatorial classes have companion actions on
their generating functions. To summarize, suppose that $\mathcal{A}$
is a combinatorial class with some notion of size, and that $A_n$ is
the number of objects of size $n$. Let $A(z)$ be the series $\sum_n
A_n z^n$. If we can express $\mathcal{A}$ in terms of other
combinatorial classes, then we can do likewise for its ogf.  For example, if $\mathcal{A}$ is the disjoint union of two
classes: $\mathcal{A}=\mathcal{B}\uplus\mathcal{C}$, the associated
generating functions satisfy the simple relation $A(z)=B(z)+C(z)$. If
$\mathcal{A}$ is described using the cartesian product and the size is
additive, $\mathcal{A}=\mathcal{B}\times\mathcal{C}=\{ (\beta,
\gamma):\beta\in\mathcal{B},\gamma\in \mathcal{C} \}$, then the
ogf satisfy $A(z)=B(z)C(z)$, the usual
product for formal power series. Finally, if class $\mathcal{A}$ is a
sequence of objects from class $\mathcal{B}$, that is, 
$\mathcal{A}=\operatorname{Seq}(\mathcal{B})=\{(\beta_1, \dots,
\beta_k): 0\leq k, \beta_i\in \mathcal{B} \}$, then there is the
generating function correspondence
\[
A(z)=\frac{1}{1-B(z)}.
\]
This is the mere surface of a vast theory rooted in the foundational
work of Chomsky and Schutzenberger and their study of algebraic
equations related to context free grammars, but significantly advanced
and summarized as the theory of decomposable structures
in~\cite{flajolet-analytic}.

This is a particularly robust formalism: we can create recursive
functional equations, and we can easily pass information about
additional parameters. We do both of these here.

\subsection{Enumeration formulas}
We easily translate the combinatorial description in
Eq.~\eqref{eqn:schroder} into the functional equation\footnote{Here we have used that $\operatorname{Seq}_{\geq
  2}(\mathcal{S})=\mathcal{S}\times\mathcal{S}\times
\operatorname{Seq}(\mathcal{S})$.}
\begin{equation}\label{eqn:schroderEq}
S(z) = z+ \frac{S(z)^2}{1-S(z)}.
\end{equation}
This converts to a simple quadratic equation in~$S(z)$. There are two
solutions and we choose the one with a Taylor series expansion at 0
with positive integer coefficients, i.e. a generating function solution. This is 
\begin{equation}\label{eqn:schroderOgf}
S(z)=\frac{z+1-\sqrt{z^2-6z+1}}{4}=\frac{z+1}{4}-\frac14\sqrt{\left(1-\frac{z}{3+\sqrt{8}}\right)\left(1-\frac{z}{3-\sqrt{8}}\right)}.
\end{equation}

In order to determine expressions for the asymptotic growth, we follow
exactly the procedure outlined in~\cite[Chapter
  VI.1]{flajolet-analytic}, in particular the flow chart
of~\cite[Figure~VI.7]{flajolet-analytic}. We outline the three main
steps of the analysis, but readers interested in further details are
referred to this resource.

The first step is to determine the dominant singularity. This is the
smallest positive real-valued singularity, which in this case is
$3-\sqrt{8}$. 

The second step is to determine the behavior of the function around
its dominant singularity, $3-\sqrt{8}$:
\begin{equation}\label{eqn:nrsing}
S(z)\sim \frac{2-\sqrt{2}}{2}-\frac12\sqrt{\sqrt{18}-4}\left(
  1-\frac{z}{3-\sqrt{8}}\right)^{\frac12} \quad\mbox{as}\quad\quad z\sim 3-\sqrt{8}.
\end{equation}

We are in a context where asymptotic transfer
theorems~\cite[VI.3]{flajolet-analytic} apply, and hence we move to
the final step. The approximation of the function near this
singularity yields an asymptotic approximation for its coefficients in
the Taylor expansion around 0. Roughly, we adapt the following
correspondence
\[
F(z)\sim\left(1-\frac{z}{\rho}\right)^{-\alpha} \quad\mbox{as}\quad z\sim \rho \implies [z^n]F(z)\sim \rho^{-n}\frac{n^{\alpha-1}}{\Gamma(\alpha)}.
\]
In this notation, $[z^n]$ extracts the coefficient of $z^n$ in the
series expansion of the expression that immediately follows, and
$\Gamma$ is the Gamma function.  Using the
approximation of $S(z)$ near its dominant singularity from
Eq.~\eqref{eqn:nrsing}, we deduce
\begin{equation}\label{eqn:numtrees}
S_n = [z^n]S(z)\sim \left(\frac{1}{4}\sqrt{\sqrt{18}-4}\right) (3-\sqrt{8})^{-n}
\frac{n^{-3/2}}{\sqrt{\pi }}\sim 0.12\,(5.88)^n\,n^{-\frac32}.
\end{equation}

An asymptotic approximation of the number $C_n$ of signed commuting permutations of size $n$ is obtained by multiplying the above equivalent by $2^{n+1}$.

\subsection{Tree parameters: A primer}
We study the average value of different tree parameters with a common
strategy, which we briefly outline here. Let
$\chi:\mathcal{S}\rightarrow \mathbb{N}$ be an non-negative
integer valued function that records some combinatorial property of a
Schr\"oder tree, such as the number of internal vertices.  The main
tool here is the bivariate generating function
\[
S(z,u)=\sum_{\tau\in\mathcal{S}} u^{\chi(\tau)}z^{|\tau|} = \sum_{k,n}
S_{k,n} u^kz^n,
\]
where $S_{k,n}$ is the number of Schr\"oder trees with $n$
leaves, with $\chi$ value equal to $k$. 
Of course, $S(z) = S(z,1)$ and $S_n = \sum_{k\geq0} S_{k,n}$.
Now, if $\mathbb{E}_n(\chi)$
is the expected value of $\chi$ over all objects of size $n$ in
$\mathcal{S}$, then by definition
\[
\mathbb{E}_n(\chi)=\frac{\sum_{k\geq0} kS_{k,n}}{\sum_{k\geq0} S_{k,n}}.
\]
We have access to this from the bivariate generating function. Remark,
\[
\frac{\partial }{\partial u} S(z,u)= \sum_{k,n} k S_{k,n} u^{k-1}z^n.
\]
hence 
\[
\mathbb{E}_n(\chi)=\frac{\left.[z^n]\frac{\partial}{\partial u} S(z,u)\right|_{u=1}}{[z^n]S(z,1)}.
\]
The denominator of this expression is calculated in
Equation~\ref{eqn:numtrees}, and in our examples the numerator is a
coefficient extraction of an algebraic function of $z$, hence the
three steps described in the previous section apply. Indeed, in our
two examples, the dominant singularity is the same as in $S(z)$, $3-\sqrt{8}$.

This is also a robust approach, and upon considering higher
derivatives we can obtain higher moments.

\subsection{The average number of internal vertices}
We begin by considering the parameter $\chi$ equal to the number of
internal vertices in a Schr\"oder tree.  We can augment the
specification in Eq.~\eqref{eqn:schroder} with a neutral
marker~$\mu$ of weight 0 to tag the internal vertices:
\begin{equation}\label{eqn:iv-schroder}
\mathcal{S} = \mathcal{L} + \mu\times\operatorname{Seq}_{\geq 2}(\mathcal{S}).
\end{equation}
We now consider the bivariate generating function $S(z,u)$ where $u$
marks~$\chi$, which counts the total number of
markers. Eq.~\eqref{eqn:iv-schroder} translates to a functional
equation for $S(z,u)$:
\[
S(z,u) = z + u \frac{S(z,u)^2}{1-S(z,u)}.
\]
This is solved in a similar manner to $S(z)$:
\[
S(z,u) =\frac{z+1-\sqrt{(z+1)^2-4z(u+1)}}{2(u+1)}=z+uz^2+(u+2u^2)z^3+\dots.
\]
This expression can be differentiated to determine
$$\frac{\partial}{\partial u}S(z,u)|_{u=1} = \frac{(z-1)^2 - (z+1)\sqrt{(z+1)^2 -8z}}{8 \sqrt{(1-\frac{z}{3+\sqrt{8}})(1-\frac{z}{3-\sqrt{8}})}}\textrm{.}$$
As we remarked above, the singularity analysis follows as before using
the singularity $3-\sqrt{8}$. Thus, 
\[
\frac{\partial}{\partial u}S(z,u)|_{u=1} \sim \alpha \cdot
\left(1-\frac{z}{3-\sqrt{8}}\right)^{-1/2}\quad\mbox{as}\quad\quad z\sim
3-\sqrt{8},  
\]
where the constant $\alpha$ is determinnd by evaluating the rest of the expression
at $z=3-\sqrt{8}$. The third step, applying the transfer theorem, is
the performed. To determine
the average, we divide this expression by the asymptotic number of
trees, as determined in Equation~\eqref{eqn:numtrees}. 
We simplify the radicals, and obtain
\begin{equation}\label{eqn:avg-iv}
\mathbb{E}_n(\chi)=\frac{\left.[z^n]\frac{\partial}{\partial u}
    S(z,u)\right|_{u=1}}{[z^n]S(z,1)}\sim
\frac{3-\sqrt{8}}{3\sqrt{2}-4}\, n\sim \frac{n}{\sqrt{2}}.
\end{equation} 

\subsection{The average number of reversals.}
Next we use the average number of internal vertices to count the
average number of reversals in a scenario. An evolutionary scenario is
obtained from tree in $\mathcal{S}$ by signing the root, and the
leaves. Each internal vertex, except the root, represents a
reversal. A leaf represents a reversal if and only if it has a sign different from
the sign of its parent.

The number of internal vertices is a good first approximation for
the number of reversals. Asymptotically, since the average number of
vertices is a linear function of $n$, subtracting by one to account
for the root has little or no effect. 

In order to account for the reversals at the leaves, we remark that
for any tree in~$\mathcal{S}$ of size~$n$, we consider all $2^n$ possible ways
of assigning signs to the leaves, and from this symmetry we deduce
that on average this adds $n/2$ reversals, all of length~$1$.

We put all of these pieces together in the following theorem.
\begin{theorem}\label{thm:numreversals}
  The asymptotic average number of reversals in a parsimonious perfect
  scenario of a random signed commuting permutation of size $n$ is
  $n/\sqrt{2}+n/2=\frac{1+\sqrt2}2\,n$. On average there are $n/2$
  reversals of length $1$.
\end{theorem}

\subsection{The average length of a reversal}
The length of a reversal in a scenario is equal to the size (number of
leaves) in the corresponding subtree of the common interval tree. 
Again, our analysis first estimates by studying a parameter on unsigned trees,
$\mathcal{S}$, and then tunes by considering the reversals of size one
represented by signs on the leaves. 

Our analysis is guided by the study of a related parameter called
pathlength that frequently makes a cameo appearance when trees are
used to study sorting algorithms.

Let $\Psi(\tau)$ be the sum of the subtree sizes for all subtrees in
$\tau\in\mathcal{S}$. Examining Figure~\ref{fig:schroder}, we can
formulate a recursive description of~$\Psi(\tau)$. Consider
$\tau\in\mathcal{S}$. Either $\tau$ is a single leaf, in which
case~$\Psi(\tau)=0$, or the root has $m$ children, labeled left to
right by $\tau_1, \dots, \tau_m$. To compute $\Psi(\tau)$, we sum the
sizes of the subtrees of each child of the root, and then add the size
of the entire tree, which of course is the sum of the sizes of the
children.  This is written
\[
\Psi(\tau)=\sum_{j=1}^m \big( \Psi(\tau_j) + |\tau_j| \big).
\]

A tree parameter that satisfies such a relationship is an
\emph{additive parameter} and writing the corresponding functional
equation is straightforward. We mark the parameter~$\Psi$ by the variable~$v$ in the bivariate
generating function $S(z,v)$:
\begin{equation}\label{eqn:sumtrees}
S(z,v)=\sum_{\tau\in\mathcal{S}} v^{\Psi(\tau)}z^{|\tau|} = z + v^2z^2 + (v^3+2v^5)z^3+\dots
\end{equation}
This parameter is identical to the pathlength parameter, and the steps
from the generating function to the equation are all well-explained
in~\cite[Section III.5]{flajolet-analytic}. We derive the functional
equation
\begin{equation}\label{eqn:sumtreeseqn}
S(z,v)=z+\frac{S(vz,v)^2}{1-S(vz, v)}.
\end{equation}

Rather than solve for $S(z,v)$, it is easier to solve for
$\frac{\partial}{\partial v} S(z,v)|_{v=1}$ directly by
differentiating Eq.~\eqref{eqn:sumtreeseqn} with respect to~$z$
and~$v$, and setting $v=1$ in the resulting equations. This leads to
two equations in two unknowns. Using the notation
$S_v(z)=\frac{\partial}{\partial v} S(z,v)|_{v=1}$,
$S_z(z)=\frac{\partial}{\partial z} S(z,v)|_{v=1}$, and recalling
$S(z,1)=S(z)$, the counting ordinary generating function for
$\mathcal{S}$, this leads to the system
\begin{align*}
S_v(z)&=\frac{S(z)(2-S(z))(S_z(z)z+S_v(z))}{(1-S(z))^2}\\
S_z(z)&=1+\frac{S(z)S_z(z)(2-S(z))}{(1-S(z))^2}.
\end{align*}
We solve $S_v(z)$ in terms of $S(z)$:
\[
S_v(z)=\frac{z S(z)(2-S(z))(1-S(z))^2}{(1-4S(z)+2S(z)^2)^2}=2z^2+13z^3+80z^4+\dots
\]
This is an explicit expression to which we apply singularity analysis
to determine an expression for the coefficient of $z^n$:
\[
[z^n]\frac{\partial}{\partial v} S(z,v)|_{v=1}=[z^n]S_v(z)\sim\frac{\sqrt{2}}{16}(3-\sqrt{8})^{-n}.
\]
This value approximates the sum of the sizes of all subtrees of all
trees in~$\mathcal{S}$.  The average value of $\Psi$ is the quotient
\begin{align*}
\mathbb{E}_n(\Psi)=\frac{[z^n]S_v(z)}{[z^n]S(z)} &\sim
\left(\frac{\sqrt{2}}{16}(3-\sqrt{8})^{-n}
\right)\left(\frac{1}{4}\sqrt{\sqrt{18}-4}
  (3-\sqrt{8})^{-n}\frac{1}{\sqrt{\pi n^3}}\right)^{-1}\\
&=\frac{\sqrt{\pi}}{4\sqrt{3-\sqrt{2}}}\,n^{\frac32} \sim 1.27\, n^{\frac32}
\end{align*}

  To get the expected sum of the lengths of the reversals of a
  parsimonious perfect scenario for a random commuting permutation, we
  consider adjustments that occur for each tree in $\mathcal{S}$, so
  we can add them directly to this value. For each tree we remove the
  size of the whole tree ($n$) since we do not count this as a
  reversal, and we also add the average contribution of the reversals
  of size 1 $(n/2)$. These two adjustments do not affect the
  asymptotic growth since $n^{\frac32}$ dominates $n$ for large $n$.

 To determine the the average length, we now divide by the average number of
 reversals, which we determined to be $\frac{1+\sqrt{2}}{2}n$ in Theorem~\ref{thm:numreversals} We summarize these results in the following theorem. 
\begin{theorem}\label{thm:avglen}
  The average length of a reversal in a parsimonious perfect
  scenario for a random signed commuting permutation of size~$n$ is
  asymptotically
\[\frac{\sqrt{\pi}}{2(1+\sqrt{2})\sqrt{3-\sqrt{2}}}\,\sqrt{n} \sim 1.054\, \sqrt{n}.\]

\end{theorem}

\section{Conclusion}
\label{sec:conclusion}

\paragraph{Summary}
Perfect sorting by reversals, although an intractable problem, is very
likely to be solved in polynomial time for random signed permutations,
under the uniform distribution. This result relies on a study of the
shape of a random strong interval tree that shows that asymptotically
such trees are mostly composed of a large prime vertex at the root and
small subtrees.  We were also able to give precise asymptotic results
for the expected lengths of a parsimonious perfect scenario and of a
reversal of such a scenario for random commuting permutations. Our
results were obtained using techniques of enumerative and analytic
combinatorics.

\paragraph{Discussion on our results.}
Recently, several works have investigated average properties of
combinatorial objects related to genomic distance computation, such as
the breakpoint graph~\cite{SLRM08,X08,XZS07}, conserved
segments~\cite{tesler-distribution} or adjacencies and common
intervals~\cite{corteel-common,XBS08}. The motivation for such works
can be twofold. One can be interested in the expected behavior of
some algorithms, such as in~\cite{SLRM08}, that shows that the most
intricate part of the theory of sorting by reversals (clearing hurdles
and fortress) is not required on uniform random permutations. Our
results on the average complexity of computing a parsimonious perfect
scenario belong to this family of results. In other cases, one can be
interested in the expected properties of an evolutionary scenario for
random genomes~\cite{X08,XZS07}. This allows, given real data,
to assess the significance of the comparison of a pair of genomes and
to compute statistical tests measuring the evolutionary signal left:
intuitively, if a scenario between two real genomes looks like a
scenario between random genomes, one can make the hypothesis that
there is little to no evolutionary signal left in the considered pair
of real genomes. Our results on commuting permutations are of such
nature.

The fact that computing a parsimonious perfect scenario requires
polynomial time on the average is mostly a theoretical result, that
completes the complexity analysis of the problem. Indeed, real data
sets (pairs of genomes or genome segments) are in general not expected
to define strong interval trees with a large number of prime nodes
(see~\cite{landau-gene} for example). So the algorithms described
in~\cite{berard-perfect,berard-more} were already known to be
efficient on real data sets.  

It should however be noted that our results on the expected shape of a
strong interval tree, and in particular on the number of prime
vertices, generalizes previous results on conserved adjacencies and
common intervals in
permutations~\cite{uno-fast,corteel-common,XBS08}. They could form the
basis for a deeper study of the expected shape of the strong interval
tree, parametrized by the number of common adjacencies or of prime
nodes. Also, the combinatorial specification of the class of strong
interval trees opens the way to random generation
algorithms~\cite{flajolet-random-generation} of trees with some prescribed
structure (such as the number or maximum degree of prime nodes), that
we outline in the paragraph below on future research. This might allow
to study by simulations the expected properties of a perfect scenario
between pairs of genomes defining strong interval trees with a
prescribed structure. Such way to assess the significance of features
of a hypothetical scenario between real genomes is clearly of
practical interest.

In the same vein, the strong interval trees obtained in comparing
pairs of mammalian genomes for example~\cite{landau-gene} contain very
few prime nodes, and then contain large subtrees that represent
commuting permutations; these subtrees can then be compared to the
expected properties of random commuting permutations to point at
genome segments whose evolution is significantly non-random.

\paragraph{Future research.}
There exists several more general models of genome rearrangements
\cite{fertin09}. Among them, the more general is based on an operation
called \emph{Double-Cut-and-Join} (DCJ for short) that models
reversals and several other types of rearrangements. The notion of
perfect DCJ scenarios has been studied in~\cite{berard-perfect2} and
has the intriguing property that instances that were hard to solve for
reversals can be solved in polynomial time in the DCJ model and
conversely. It would then be interesting to compare the average time
complexity of perfect sorting by DCJ to the results we describe in the
present work.

We could, modulo the labeling of the prime nodes by simple
permutations, easily describe ${\cal T}_n$ using a grammar in the
combinatorial calculus described in~\cite{flajolet-analytic}. This
would give access to enumerative and structural information, when
paired with the generating function for simple
permutations. Generating these trees directly, i.e. without first
generating the corresponding permutation, remains an interesting open
problem, that seems to be well suited to Boltzmann random generation
techniques~\cite{duchon-boltzmann}.

Also, PQ-trees are a natural family of trees that are both related to
common intervals of permutations~\cite{bergeron-computing} and used in
comparative genomics~\cite{landau-gene}. Investigating average
properties of PQ-trees is a natural extension of the work presented
here. 

More generally, average properties of the many families of
combinatorial objects that appear in comparative genomics models and
algorithms is an almost completely open field, that contains many
challenging problems and deserve being investigated.


\section{Acknowledgments} Authors MM and CC acknowledge funding
support from the Discovery Grant program of the Natural Science and
Engineering Research Council (Canada). MB and DR are supported by the
ANR projects GAMMA (BLAN07-2\_195422) and MAGNUM
(2010\_BLAN\_0204). 


\end{document}